\documentclass{llncs}
\pdfoutput=1
\usepackage{amsmath,amssymb}
\usepackage{graphicx}
\usepackage{color}

\newcommand{\dimdm}{\ensuremath{\mathrm{dim}_{\mathrm{DM}}}}
\newcommand{\RR}{\ensuremath{\mathbb{R}}}
\newcounter{obs}
\usepackage{color,url}

\newcommand{\dan}[1]{\textcolor{blue}{#1}}

\title{Dushnik-Miller dimension of $d$-dimensional tilings with boxes}
\author{Mathew C. Francis~\inst{1} and Daniel Gon\c{c}alves~\inst{2}}
\institute{Indian Statistical Institute, Chennai Centre, India \and
  LIRMM, Univ. de Montpellier \& CNRS, France.}

\begin{document}
\maketitle

\begin{abstract}
Planar graphs are the graphs with Dushnik-Miller dimension at most
three (W. Schnyder, Planar graphs and poset dimension, Order 5,
323-343, 1989).  Consider the intersection graph of interior disjoint
axis parallel rectangles in the plane. It is known that if at most
three rectangles intersect on a point, then this intersection graph is
planar, that is it has Dushnik-Miller dimension at most three.  This
paper aims at generalizing this from the plane to $\mathbb{R}^d$ by
considering tilings of $\mathbb{R}^d$ with axis parallel boxes, where
at most $d+1$ boxes intersect on a point. Such tilings induce
simplicial complexes and we will show that those simplicial complexes
have Dushnik-Miller dimension at most $d+1$.
\end{abstract}


\section{Introduction}

One can easily see that the intersection graph induced by a set of
interior disjoint axis parallel rectangles, with at most three
rectangles intersecting on a point, is a planar graph. C. Thomassen
characterized those graphs~\cite{Tho}
(See also~\cite{Fusy} for a combinatorial study of these
representations). H. Zhang showed how such a representation (when it tiles a
rectangle) also induces a Schnyder wood of the induced planar
graph~\cite{Z10}. Schnyder woods was the key structure that allowed
W. Schnyder to prove that planar graphs are the graphs with
Dushnik-Miller dimension at most three~\cite{S89}. It is interesting
to note that most planar graphs have Dushnik-Miller dimension equal to
three. Indeed, a graph has Dushnik-Miller dimension at most two if and
only if it is the subgraph of a path.

The main result of this paper is that the simplicial complexes induced
by a wide family of tilings of $\mathbb{R}^d$ with axis parallel
boxes, have Dushnik-Miller dimension at most $d+1$. As most of these
simplicial complexes have a $d$-face, there Dushnik-Miller dimension
is exactly $d+1$. Definitions are provided in the following. Note that
both, the objects (graphs or simplicial complexes) with Dushnik-Miller
dimension greater than three~\cite{FK,GI17,Oss}, and the systems of
interior disjoint axis parallel boxes in
$\mathbb{R}^d$~\cite{EuroCG,FF,AA10}, are difficult to handle but are
raising interest in the community. The proof of our result generalizes
H. Zhang's idea for constructing Schnyder woods in tilings of
$\mathbb{R}^2$, and relies on some properties of tilings that are of
independent interest. After providing a few basic definitions in
Section~\ref{sec:d-boxes}, we present these properties in
Section~\ref{sec:d-tiling} and Section~\ref{sec:sep}. We then prove
our main result in Section~\ref{sec:dimdm}. We then conclude with some
open problems.

\section{$d$-boxes}\label{sec:d-boxes}

A \emph{$d$-box} (or simply a \emph{box}) is the Cartesian product of
$d$ closed intervals. The intervals of a $d$-box $B$ are denoted
$B_1,\ldots,B_d$, that is $B=B_1\times\cdots\times B_d$.  Similarly
the coordinates of a point $x\in \mathbb{R}^d$ are denoted
$x=(x_1,\ldots,x_d)$.  The endpoints of an interval $B_i$ are denoted
$B_i^-$ and $B_i^+$ in such a way that $B_i = [B_i^-,B_i^+]$. Such an
interval is \emph{degenerate} if its endpoints coincide, that is if
$B_i^- = B_i^+$.  The {\em dimension} $\dim(B)$ of a $d$-box $B$ is
the number of non-degenerate intervals among $\{B_1,\ldots,B_d\}$. For
example, a 0-dimensional $d$-box is a point in $\mathbb{R}^d$.


The \emph{interior} of a $d'$-dimensional box $B$ is the open
$d'$-dimensional box defined by the points $p=(p_1,\ldots,p_d)$ such
that $B_i^- < p_i < B_i^+$ if $B_i^-\neq B_i^+$ or such that
$p_i=B_i^- = B_i^+$ otherwise.  The points of $B$ that are not
interior form the \emph{border} of $B$. The border of $B$ is the union
of its sides. A \emph{side} of $B$ is a $(d'-1)$-dimensional box
$S(B,i,*) = [B_1^-,B_1^+]\times\cdots\times
[B_i^*,B_i^*]\times\cdots\times [B_d^-,B_d^+]$, for $*\in\{-,+\}$, and
some non-degenerate dimension $i$ of $B$ (i.e. such that $B_i^-\neq
B_i^+$). Clearly, a box with dimension $d'$ has $2d'$ distinct
sides.  A \emph{corner} of $B$ is a point $(x_1,x_2,\ldots,x_d)$ where
each $x_i$ is an endpoint of $B_i$, that is either $x_i=B_i^-$ or
$x_i=B_i^+$. Clearly, a box with dimension $d'$ has $2^{d'}$
corners.

The intersection $A\cap B$ of two boxes is the box $(A_1\cap
B_1)\times\cdots\times (A_d\cap B_d)$. Two $d$-dimensional boxes
are \emph{interior disjoint} if their interiors do not intersect, or
equivalently if $\dim(A\cap B)<d$.

\begin{definition}
A \emph{$d$-tiling} $\mathcal{T}$ is a collection of interior disjoint
$d$-dimensional boxes contained in $[-1,+1]^d$ that tile
$[-1,+1]^d$ (i.e. every point of $[-1,+1]^d$ belongs to at least one
box of $\mathcal{T}$).
\end{definition}

Let us now define $\mathcal{T}_{ext}=\{T(i,*)\ :\ 1\le i\le
d,\ *\in\{-,+\} \}$, a set of $2d$ $d$-dimensional boxes that tile
$\mathbb{R}^d \setminus [-1,+1]^d$.
\begin{center}
\begin{tabular}{ccccccc}
$T(i,-)_j $ & $\ =\ $ & $T(i,+)_j$ & $\ =\ $ & $[-1,+1]$ & \mbox{ if } & $\ j<i$\\
& & $T(i,-)_j$ & $\ =\ $ & $[-\infty,-1]$ & \mbox{ if } & $\ j=i$\\
& & $T(i,+)_j$ & $\ =\ $ & $[+1,+\infty]$ & \mbox{ if } & $\ j=i$\\
$T(i,-)_j$ & $\ =\ $ & $T(i,+)_j$ & $\ =\ $ & $[-\infty,+\infty]$ & \mbox{ if } & $\ j>i$
\end{tabular}
\end{center}
In particular,

\begin{tabular}{ccccccccc}
$T(1,-)$ & $=$ & $[-\infty,-1]$ & $\times$ & $[-\infty,+\infty]$ & $\times$ & $[-\infty,+\infty]$ & $\times$ & $\cdots$\\

$T(1,+)$ & $=$ & $[+1,+\infty]$ & $\times$ & $[-\infty,+\infty]$ & $\times$ & $[-\infty,+\infty]$ & $\times$ & $\cdots$\\

$T(2,-)$ & $=$ & $[-1,+1]$ & $\times$ & $[-\infty,-1]$ & $\times$ & $[-\infty,+\infty]$ & $\times$ & $\cdots$\\

$T(2,+)$ & $=$ & $[-1,+1]$ & $\times$ & $[+1,+\infty]$ & $\times$  & $[-\infty,+\infty]$ & $\times$ & $\cdots$\\

$\vdots$ & & & & & & & & \\
\end{tabular}

Note that given a $d$-tiling $\mathcal{T}$, the set
$\mathcal{T}\cup\mathcal{T}_{ext}$ is a set of interior disjoint
$d$-dimensional boxes that tile $\mathbb{R}^d$. The set
$\mathcal{T}_{ext}$ is needed to define \emph{proper} $d$-tilings in
Section~\ref{sec:d-tiling}, and it is used in the following technical
lemma.

Given two intersecting boxes $A$ and $B$, if $A_i\cap B_i$ is
degenerate, then these two boxes are said to \emph{touch} in
dimension $i$.

\begin{lemma}\label{lem:touch}
In a $d$-tiling $\mathcal{T}$, for any $A \in \mathcal{T}$ and any
point $p\in A$. If $p_i=A_i^-$ (resp. $p_i=A_i^+$), there is a box
$B\in \mathcal{T}\cup\mathcal{T}_{ext}$ such that $p\in A\cap B$, and
such that $A$ and $B$ touch only in dimension $i$ (i.e. such that
$A_j\cap B_j$ is degenerate only for $j=i$). In particular,
$\dim(A\cap B)= d-1$.
\end{lemma}

\begin{proof}
Define $Z\subseteq\mathbb{R}$ to be the set of all distinct numbers that
appear as a coordinate of some corner of some box in $\mathcal{T}\cup
\mathcal{T}_{ext}$, i.e., $$Z=\bigcup_{B\in\mathcal{T}\cup
\mathcal{T}_{ext}}\bigcup_{j\in\{1,\ldots,d\}} \{B^-_j,B^+_j\}$$
Choose a real number $\epsilon>0$ such that $\epsilon<\min\{|a-b|\colon a,b\in
Z, a\neq b\}$.
We now choose a point $q\in\mathbb{R}^d$ such that for each $j\in\{1,\ldots,
d\}\setminus\{i\}$, $q_j\in A_j$ and $q_i\notin A_i$. We also make sure that
$q$ is so close to $p$ that any box $B$ that contains $q$ also contains
$p$. This can be achieved by choosing $q$ as follows.
$$q_j=\left\{
\begin{array}{ll}
p_j-\epsilon&\mbox{ if }j=i\mbox{ and }p_j=A^-_j\\
p_j+\epsilon&\mbox{ if }j=i\mbox{ and }p_j=A^+_j\\
p_j&\mbox{ if }A^-_j<p_j<A^+_j\\
p_j+\epsilon&\mbox{ if }j\neq i\mbox{ and }p_j=A^-_j\\
p_j-\epsilon&\mbox{ if }j\neq i\mbox{ and }p_j=A^+_j
\end{array}\right.$$
Clearly, there is some box $B\in\mathcal{T}\cup\mathcal{T}_{ext}$ such
that $q\in B$. Suppose that $p\notin B$, i.e., there exists some
$j\in\{1, \ldots,d\}$ such that $p_j\notin B_j$. Then since $q_j\in
B_j$, we have either $p_j<B^-_j\leq q_j$ or $q_j\leq B^+_j<p_j$. From
the definition of $q$, we have that $p_j\in\{A^+_j,A^-_j\}$ (as
otherwise $q_j=p_j$) and that $|p_j-q_j|= \epsilon$. This means that
there exist distinct $a,b\in\{B^+_j,B^-_j,A^+_j, A^-_j\}$ such that
$|a-b|\dan{\le} \epsilon$, which contradicts our choice of
$\epsilon$. Therefore, we conclude that $p\in B$, and so $A$ and $B$ intersect.

Actually for each $j\in\{1,\ldots,d\}\setminus\{i\}$, by construction
$q_j\in A_j$ and $q_j \notin\{A^-_j,A^+_j\}$. As $\{p_j,q_j\} \subset
A_j\cap B_j$, we thus have that $A_j\cap B_j$ is not degenerate. As
$A$ and $B$ intersect on a box of dimension $d'<d$, it must be the
case that $A_i\cap B_i$ is degenerate.  This completes the
proof.\hfill\qed
\end{proof}

Let $\mathcal{H}^{(i)}_{x}$ denote the hyperplane defined by
$\left\{p\in \mathbb{R}^d\ :\ p_i=x \right\}$. Given a
$d'$-dimensional box $B$, let us denote $B|^{(i)}_x$ the intersection
between $B$ and $\mathcal{H}^{(i)}_{x}$.  Note that this intersection
is either empty, or the box $B$ itself (if $B_i=[x,x]$), or it is a
$(d'-1)$-dimensional box. Depending on the context, $B|^{(i)}_x$ is
considered as a box of $\mathbb{R}^d$, or of $\mathbb{R}^{d-1}$ if we
omit the $i$-th dimension.  The hyperplane $\mathcal{H}^{(i)}_{x}$ is
said to be {\em generic} w.r.t. a $d$-tiling $\mathcal{T}$, if for
every $B\in \mathcal{T}$, $x\neq B_i^-$ and $x\neq B_i^+$.  Given a
$d$-tiling $\mathcal{T}$, let the \emph{intersection} of $\mathcal{T}$
and $\mathcal{H}^{(i)}_{x}$ be the set $\mathcal{T}|^{(i)}_{x} =
\left\{ B|^{(i)}_x \ :\ B\in\mathcal{T} \text{ with } B^-_i \le x \le
B^+_i \right\}$. The following lemma indicates when
$\mathcal{T}|^{(i)}_{x}$ defines a $(d-1)$-tiling.

\begin{lemma}~\label{lem:slice-tiling}
  For every $d$-tiling $\mathcal{T}$, and every hyperplane
  $\mathcal{H}^{(i)}_{x}$ that is generic w.r.t. $\mathcal{T}$,
  $\mathcal{T}|^{(i)}_{x}$ is a $(d-1)$-tiling.
\end{lemma}
\begin{proof}
As the boxes in $\mathcal{T}$ are $d$-dimensional, it is clear that
$\mathcal{T}|^{(i)}_{x}$ is a set of $(d-1)$-dimensional $(d-1)$-boxes
(by omitting the $i$-th dimension). It is also clear that these
boxes span $[-1,+1]^{d-1}$. What remains to prove is that these boxes
are interior disjoint.

Towards a contradiction, suppose that $A'$ and
$B'\in\mathcal{T}|^{(i)}_{x}$ intersect on a $(d-1)$-dimensional
$(d-1)$-box $C'$. Thus if $A$ and $B$ are the corresponding original boxes in
$\mathcal{T}$, i.e. those such that $A'=A|^{(i)}_{x}$ and $B'=B|^{(i)}_{x}$,
then $C'_j=A_j\cap B_j$ is non-empty and non-degenerate for each dimension $j$
other than $i$. Moreover, $A_i\cap B_i$ is also non-empty as $x\in A_i\cap B_i$.
As $\mathcal{T}$ is a $d$-tiling, this means that $A_i\cap B_i$ is degenerate, or
in other words, $A$ and $B$ touch in dimension $i$. This implies that $x=A^-_i$
or $x=A^+_i$, a contradiction to the fact that $\mathcal{H}^{(i)}_{x}$ is
  generic w.r.t. $\mathcal{T}$.\hfill\qed
\end{proof}

For any $d$-tiling $\mathcal{T}$ and any hyperplane
$\mathcal{H}^{(i)}_{x}$ that is generic w.r.t. $\mathcal{T}$, let us define $\mathcal{T}|^{(i)-}_{x}$ and
$\mathcal{T}|^{(i)+}_{x}$ as the two parts obtained by cutting
$\mathcal{T}$ through $\mathcal{H}^{(i)}_{x}$. In order to obtain
$d$-tilings, we prolong the sides on $\mathcal{H}^{(i)}_{x}$
towards $\mathcal{H}^{(i)}_{+1}$, or towards
$\mathcal{H}^{(i)}_{-1}$. Formally, for any $B\in
\mathcal{T}$, let $B|^{(i)-}_{x}$ and $B|^{(i)+}_{x}$ be the boxes $B_1
\times \ldots \times \alpha^-(B_i)\times \ldots \times B_d$, and $B_1
\times \ldots \times \alpha^+(B_i)\times \ldots \times B_d$ where,
$$
  \alpha^-(B_i)=
      \begin{cases}
      \emptyset & \text{if $x\le B_i^-$}\\
      B_i & \text{if $B_i^+ < x$}\\
      [B_i^-,+1] & \text{otherwise}
    \end{cases}       
\ \ \ \ \text{and}\ \ \ \
  \alpha^+(B_i)=
      \begin{cases}
      \emptyset & \text{if $B_i^+\le x$}\\
      B_i & \text{if $x<B_i^-$}\\
      [-1,B_i^+] & \text{otherwise}
    \end{cases}       
$$

Now, let $\mathcal{T}|^{(i)-}_{x}$ (resp. $\mathcal{T}|^{(i)+}_{x}$) be
the set of non-empty boxes of the form $B|^{(i)-}_{x}$
(resp. $B|^{(i)+}_{x}$) for each $B\in \mathcal{T}$.
The following lemma is trivial.
\begin{lemma}\label{lem:cut-tiling}
  For every $d$-tiling $\mathcal{T}$, and every hyperplane
  $\mathcal{H}^{(i)}_{x}$ that is generic w.r.t. $\mathcal{T}$,
  both $\mathcal{T}|^{(i)-}_{x}$ and $\mathcal{T}|^{(i)+}_{x}$ are
  $d$-tilings.
\end{lemma}

\section{Proper $d$-tilings}\label{sec:d-tiling}

The boxes satisfy the Helly property. Indeed, given a set
$\mathcal{B}$ of pairwise intersecting boxes, the set $\cap_{B\in
  \mathcal{B}} B$ is a non-empty box. Graham-Pollak's
Theorem~\cite{Graham-Pollak,Tverberg} asserts that to partition the
edges of $K_n$ into complete bipartite graphs, one needs at least
$n-1$ such graphs. Using this theorem Zaks proved the following.

\begin{lemma}[\cite{Z85} Zaks 1985]
\label{lemma-d-1-proper}
Consider a set $\mathcal{B}$ of $d$-dimensional boxes that
intersect, that is $\bigcap_{B\in \mathcal{B}} B \neq \emptyset$. If
for every pair $A,B\in \mathcal{B}$, $\dim(A\cap B)=d-1$, then
$|\mathcal{B}|\leq d+1$.
\end{lemma}
We include a proof of this result for completeness.
\begin{proof}
Consider a point $x\in \bigcap_{B\in \mathcal{B}} B$. For any two
boxes $A, B\in \mathcal{B}$, since $\dim(A\cap B)=d-1$, there
exists exactly one dimension in which they touch. If this dimension is
$t$, then $A_t\cap B_t=\{x_t\}$. As $A_t$ and $B_t$ are
non-degenerate, we either have $A^+_t = x_t = B^{-}_t$ or $B^{+}_t =
x_t = A^{-}_t$.

Let $K$ be the complete graph with vertex set $\mathcal{B}$.  Now
label each edge $AB$ of $K$ with $t$ if $A$ and $B$ touch in dimension
$t$. As every pair of boxes touch in exactely one dimension, this
labeling defines an edge partition of $K$ into $d$ subgraphs
$G_1,\ldots,G_d$. Let us now prove that every such graph $G_i$ is a
complete bipartite graph. The vertices $A$ with an incident edge in
$G_t$ divide into two categories, those such that $x_t=A^{-}_t$ and
those such that $x_t=A^{+}_t$. Any two boxes in the same category do
not touch in dimension $t$, so these categories induce two independent
sets in $G_t$. On the other hand, any two boxes in different
categories do touch in dimension $t$, so they are adjacent in $G_t$.

So, by Graham-Pollak's Theorem, $k \le d+1$.\hfill\qed
\end{proof}


A $d$-tiling $\mathcal{T}$ is \emph{proper} if every point $p\in
\mathbb{R}^d$ is contained in at most $d+1$ boxes of
$\mathcal{T}\cup\mathcal{T}_{ext}$. Figure~\ref{fig:framed_and_proper}
provides two configurations that are forbidden in proper $d$-tilings.

\begin{figure}[h]
  \centering \includegraphics[scale=0.25]{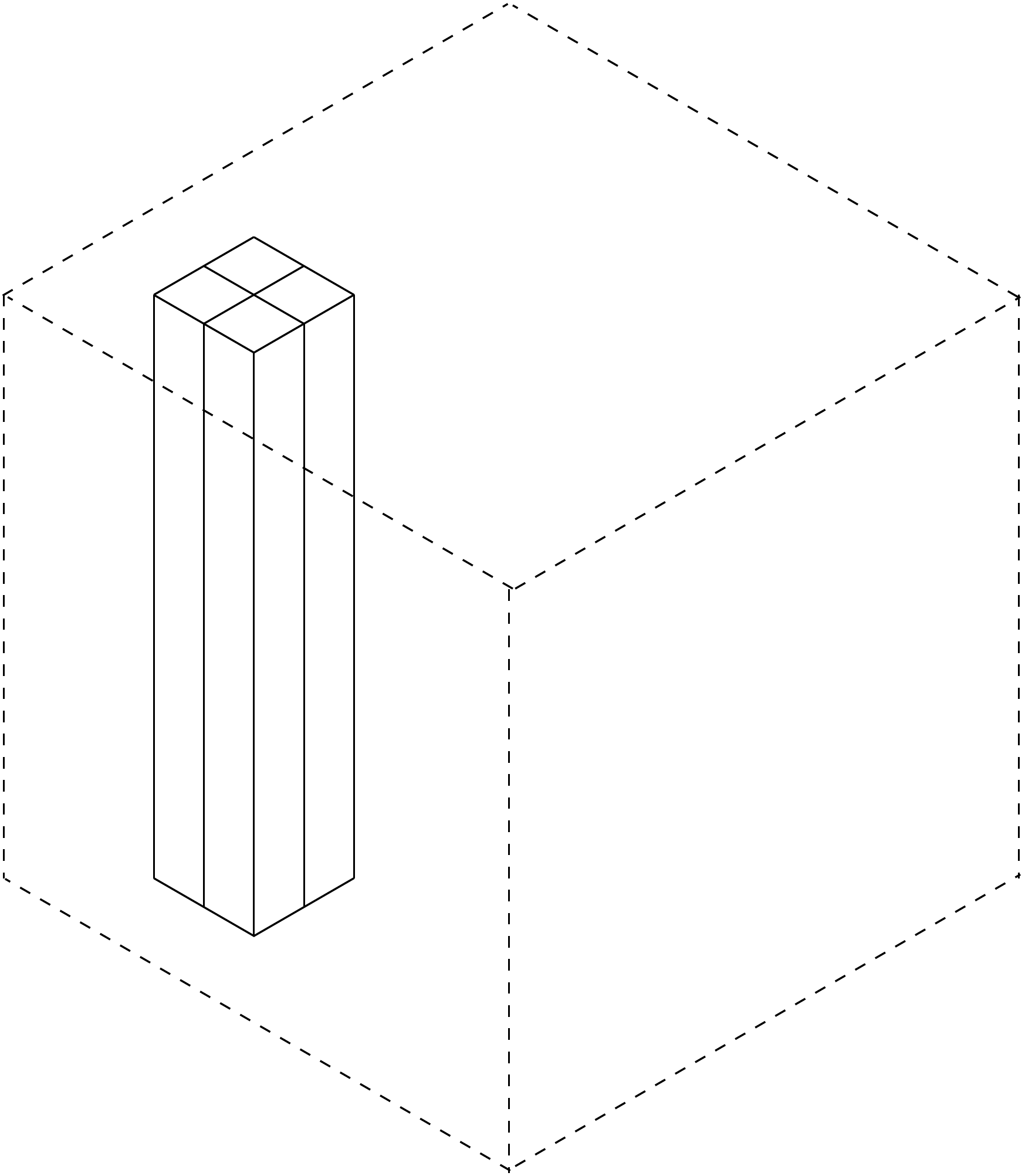}
  \hspace{3cm}
  \includegraphics[scale=0.25]{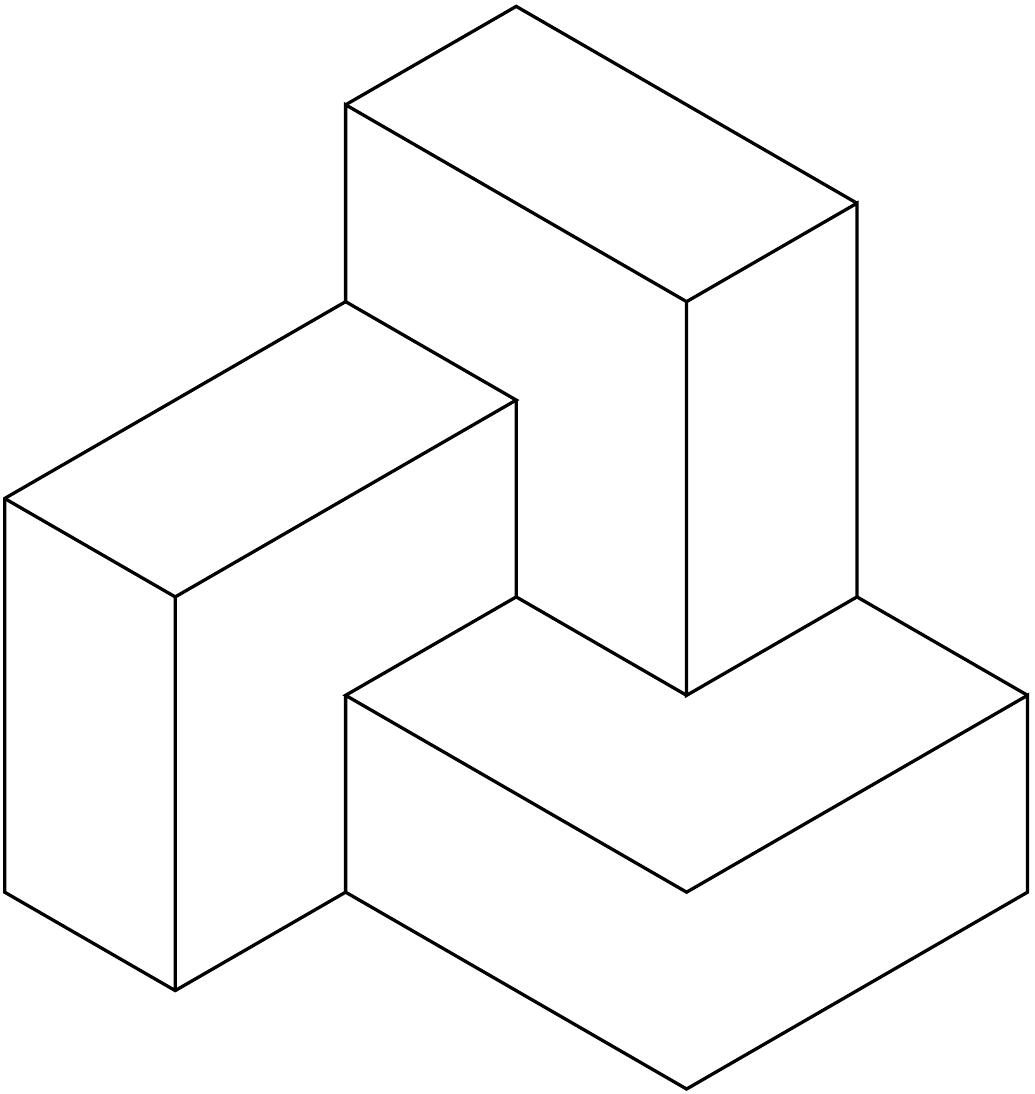}
  \caption{(left) A configuration that can appear in a tiling of
    $[-1,+1]^3$. Here every point of $[-1,+1]^3$ belongs to at most 4
    boxes of $\mathcal{T}$, but it is not proper when considering
    $\mathcal{T}_{ext}$ as two points would belong to 5 boxes.
    (right) A configuration of 3-boxes pairewise intersecting on a
    2-dimensional box. If this configuration was part of a 3-tiling
    there would be 5 boxes intersecting on a point, and thus it would
    not be proper.
  }\label{fig:framed_and_proper}
\end{figure}

\begin{lemma}~\label{lem:prop=d-1}
In a proper $d$-tiling $\mathcal{T}$, every pair of intersecting boxes
intersect on a $(d-1)$-dimensional box.
\end{lemma}
\begin{proof}
We shall prove the following stronger statement:

($*$) In a $d$-tiling $\mathcal{T}$, if two boxes $A, B \in
\mathcal{T} \cup \mathcal{T}_{ext}$ touch in more than one dimension,
then there exists a point in $A \cap B$ that is contained in at least
$d+2$ boxes of $\mathcal{T}\cup\mathcal{T}_{ext}$.

For the sake of contradiction, consider two boxes $A$ and $B$ such
that $\dim(A\cap B)=d-s$, with $s\ge 2$. Without loss of generality,
we assume that these boxes touch in dimension $i$ for $1\leq i\leq
s$. Furthermore, we assume that $A_i^+ = B_i^-$ and let us call $p_i$
this value, for $1\leq i\leq s$.

Clearly, $A$ and $B$ do not belong both to $\mathcal{T}_{ext}$. Furthermore,
as any box of $\mathcal{T}$ is contained in $[-1,+1]^d$,
it can touch a box of $\mathcal{T}_{ext}$ in at most one
dimension. Therefore, we have that $A$ and $B\in \mathcal{T}$.

We shall prove ($*$) by induction on $d-s$.  As the base case, we
shall show that it is true when $d-s=0$.  We claim that the point
$p=(p_1,\ldots,p_d)$ is contained in at least $d+2$ boxes.
By Lemma~\ref{lem:touch}, there is a box $H^{(i)}\neq A$ in
$\mathcal{T}\cup\mathcal{T}_{ext}$ such that $H^{(i)}$ touches $A$
only in dimension $i$ and contains the point $p$.  As $H^{(i)}$
touches $A$ only in dimension $i$, all these boxes are distinct.
Furthermore, as $\dim(A\cap H^{(i)}) = d-1$, each $H^{(i)}$ is
different from $B$. So, together with $A$ and $B$ they form a
collection of $d+2$ boxes that contain the point $p$.

We consider now the case $d-s>0$. As $A_d\cap B_d$ is non degenerate,
let us assume without loss of generality that $B_d^- < A_d^+ \le
B_d^+$. For an arbitrarily small $\epsilon >0$, we have that no box of
$\mathcal{T}$ has $x=A_d^+ -\epsilon$ as an endpoint of its $d^{st}$
interval.

By Lemma~\ref{lem:slice-tiling}, $\mathcal{T}|^{(d)}_{x}$ is a
$(d-1)$-tiling. Note that this tiling contains $A$ and $B$, and that
those still touch in $s\ge 2$ dimensions. Therefore, by the induction
hypothesis, there exists a point $p'=(p'_1,\ldots,p'_{d-1})$ that is
contained in at least $d+1$ of these $(d-1)$-boxes and also such that
$p'\in A\cap B$. Let these $(d-1)$-boxes be
$A,B,H^{(1)},\ldots,H^{(d-1)}$. Coming back to $\mathbb{R}^d$, the
point $p=(p'_1,\ldots,p'_{d-1}, A_d^+)$ is contained in $A\cap B\cap
H^{(1)}\cap\cdots\cap H^{(d-1)}$. By Lemma~\ref{lem:touch}, there
exists a box $F$ that contains the point $p$ and touches $A$ in
dimension $d$. This means that $F_d^- = A_d^+ > x$,
therefore $F\not\in\{A,B,H^{(1)},\ldots,H^{(d-1)}\}$. The point $p$ is
thus contained in $d+2$ distinct boxes. This concludes the proof
of the lemma.\hfill\qed
\end{proof}

Lemma~\ref{lemma-d-1-proper} and Lemma~\ref{lem:prop=d-1} imply the
following.
\begin{theorem}\label{thm:contactsystem}
A $d$-tiling is proper if and only if for any two intersecting boxes
$A,B$ we have $\dim(A\cap B)=d-1$.
\end{theorem}

Lemma~\ref{lem:prop=d-1} also allows us to prove the following
improvement of Lemma~\ref{lem:slice-tiling}.
\begin{lemma}~\label{lem:slice-tiling-proper}
  For every proper $d$-tiling $\mathcal{T}$, and every hyperplane
  $\mathcal{H}^{(i)}_{x}$ that is generic w.r.t. $\mathcal{T}$,
  $\mathcal{T}|^{(i)}_{x}$ is a proper $(d-1)$-tiling.
\end{lemma}
\begin{proof}
By Lemma~\ref{lem:slice-tiling}, $\mathcal{T}|^{(i)}_{x}$ is a
$(d-1)$-tiling. It remains to prove that it is proper.  For any pair
of intersecting boxes $A|^{(i)}_{x}$ and $B|^{(i)}_{x}\in
\mathcal{T}|^{(i)}_{x}$, Theorem~\ref{thm:contactsystem} implies that
the original boxes $A,B\in \mathcal{T}$ touch in exactly one
dimension. As $x$ is not an endpoint of $A_{i}$ or $B_{i}$, this
dimension cannot be $i$. So the boxes $A|^{(i)}_{x}$ and
$B|^{(i)}_{x}$ touch in exactly one dimension. By
Theorem~\ref{thm:contactsystem}, this implies that
$\mathcal{T}|^{(i)}_{x}$ is proper.\hfill\qed
\end{proof}

Theorem~\ref{thm:contactsystem} provides us a simple proof for the
following strengthening of Lemma~\ref{lem:cut-tiling}.

\begin{lemma}~\label{lem:cut-tiling-proper}
  For every proper $d$-tiling $\mathcal{T}$, and every hyperplane
  $\mathcal{H}^{(i)}_{x}$ that is generic w.r.t. $\mathcal{T}$,
  both $\mathcal{T}|^{(i)-}_{x}$ and $\mathcal{T}|^{(i)+}_{x}$ are
  proper $d$-tilings.
\end{lemma}
\begin{proof}
By Theorem~\ref{thm:contactsystem}, we can suppose towards a
contradiction that there are two boxes $A$ and $B$ of
$\mathcal{T}|^{(i)-}_{x} \cup \mathcal{T}_{ext}$ intersecting on a
$d'$-dimensional box for some $d'<d-1$. Clearly, $A$ and $B$ cannot
belong to $\mathcal{T}_{ext}$ both.  If $A\in \mathcal{T}|^{(i)-}_{x}$
and $B\in\mathcal{T}_{ext}$, say that $B=T(i,*)$ for some $i\in \{1,\ldots,d\}$
and $*\in\{-,+\}$, then as $A_j \subseteq [-1,+1] \subseteq B_j$ for
all $j\neq i$ we have that $d'=d-1$, a contradiction. Finally if both
$A$ and $B$ belong to $\mathcal{T}|^{(i)-}_{x}$, and if they
respectively come from $A'$ and $B'$ of $\mathcal{T}$, then we have
that $A_j\cap B_j = A'_j\cap B'_j$ for all $j\neq i$. Furthermore, by
construction $A_i\cap B_i$ is a non-degenerate interval if and only if
$A'_i\cap B'_i$ is non-degenerate.  Thus $A'\cap B'$ is also
$d'$-dimensional, a contradition.\hfill\qed
\end{proof}

Lemma~\ref{lem:slice-tiling-proper} now allows us to strengthen
Theorem~\ref{thm:contactsystem} as follows.

\begin{theorem}~\label{thm:prop=d+1-k}
A $d$-tiling is proper, if and only if for any set
$\mathcal{B}$ of pairwise intersecting boxes we have
$\dim(\cap_{B\in \mathcal{B}} B)=d+1-|\mathcal{B}|$.
\end{theorem}
\begin{proof}
Theorem~\ref{thm:contactsystem} clearly implies that this condition is
sufficient. Let us then prove that this condition is necessary. We
have to prove that in a proper $d$-tiling, for any set $\mathcal{B}$
of pairwise intersecting boxes we have $\dim(\cap_{B\in \mathcal{B}}
B)=d+1-|\mathcal{B}|$.

Let $k=|\mathcal{B}|$. By Lemma~\ref{lemma-d-1-proper} we know that
$1\le k \le d+1$. We already know that the implication holds for $k=1$
or $k=2$. For the remaining cases we proceed by induction on the pair
$(d,k)$. That is, we assume the theorem holds for $(d',k')$ ($k'$
pairwise intersecting boxes in a proper $d'$-tiling) with $d'<d$ or with
$d'=d$ and $k'<k$.

Consider any set $\mathcal{B}$ of $k$ pairwise intersecting boxes, and
let $I=\cap_{B\in \mathcal{B}} B$.

{\bf We consider first that $\dim(I)>0$}, and we assume without loss of
generality that $I$ is non degenerate in dimension $d$. Let $x$ be a
value in the interior of $I_d$, and such that $x$ is not an endpoint
of $A_d$, for any box $A\in \mathcal{T}$. By
Lemma~\ref{lem:slice-tiling-proper} $\mathcal{T}|^{(d)}_x$ is a proper
$(d-1)$-tiling. In this tiling, $\mathcal{B}$ is also a set of $k$
pairwise intersecting boxes, which intersect on $I|^{[d)}_x$. By
  induction on $(d-1,k)$ we have that $\dim(I|^{[d)}_x) = d-k$, and
    thus that $\dim(I) = d+1-k$.

{\bf We now consider that $\dim(I)=0$} and that $k\le d$ (if $k=d+1$
we are fine).  Consider any $B\in \mathcal{B}$, and let $J=\cap_{H\in
  \mathcal{B} \setminus B} H$. By induction hypothesis
$\dim(J)=d+2-k\ge 2$, and without loss of generality we consider that
$J$ is non-degenerate in dimension $i$ if and only if $1\le i\le
\dim(J)$. As $I=B\cap J$ we have that $\dim(B\cap J)=0$ and we can
assume without loss of generality that $B_i^+ = J_i^-$ for every
$i\in\{1,\ldots, \dim(J)\}$. Let us also denote by $p$ the point where
$B$ and $J$ intersect.

\begin{claim}
For every $i\in\{1,\ldots,\dim(J)\}$ there is a box $F$ of
$(\mathcal{T}\cup \mathcal{T}_{ext})\setminus \mathcal{B}$ such that
$p\in F$, such that $\dim(F \cap J ) = \dim(J)-1 = d+1-k$, and such
that $F_i\cap J_i$ is degenerate.
\end{claim}
For each $i\in \{1,\ldots ,\dim(J)\}$, if $J_i^-=p_i$ it is because
some box $A\in \mathcal{B} \setminus B$ is such that $A^-_i= p_i$. Let
$q$ be an interior point of $J|^{(i)}_{p_i}$ that is arbitrarily close
to $p$ (thus every box containing $q$ also contains $p$).  By
Lemma~\ref{lem:touch} there is a box $F\in \mathcal{T}\cup
\mathcal{T}_{ext}$ such that $q\in A\cap F$, and such that $A$ and $F$
touch only in dimension $i$.  As $F^+_i= p_i$, we have that $F_i\cap
J_i = [p_i,p_i]$ is degenerate while for every $j\in \{1,\ldots
,\dim(J)\}\setminus \{i\}$ as $q_j\in F_i\cap J_i = F_i \cap_{H\in
  \mathcal{B} \setminus B} H $ the interval $F_i\cap J_i$ is
non-degenerate.  We thus have that $\dim(F\cap J) = \dim(J)-1\ge 1$.
As $\dim(B\cap J)=0$ we have that $F\neq B$, and as $J\not\subset F$,
$F\notin \mathcal{B} \setminus B$.  So this box $F$ does not belong to
$\mathcal{B}$.

\begin{claim}
There are $\dim(J)$ such boxes $F$.
\end{claim}
If a box $F$ is such that $F_i\cap J_i$ is degenerate for two distinct
values in $\{1,\ldots ,\dim(J)\}$ then $\dim(F\cap J) \le \dim(J)-2$,
and so $F$ does not verifies the previous claim. So for each value
$i\in \{1,\ldots ,\dim(J)\}$ there is a distinct box $F$ fulfilling
the previous claim.

The theorem now follows from the fact that all the $k$ boxes of
$\mathcal{B}$ and all the $\dim(J)=d+2-k$ boxes $F$ intersect at
$p$, contradicting the fact that $\mathcal{T}$ is proper.\hfill\qed
\end{proof}





\section{Separations}\label{sec:sep}

Let us now define an equivalence relation $\sim$ on the set of sides
of $\mathcal{T}\cup \mathcal{T}_{ext}$.  The relation $\sim$ is the
transitive closure of the relation linking two sides if and only if
they intersect on a $(d-1)$-dimensional box. If the boxes $A$ and $B$
touch in dimension $i$, then $S(A,i,*)$ and $S(B,i,*^{-1})$ intersect
on a $(d-1)$-dimensional box, for some $*\in\{-1,+1\}$.  A
\emph{separation} is then defined as the union of all the boxes of
some equivalence class of $\sim$. Note that a separation is a finite
union of $(d-1)$-dimensional boxes that are degenerate in the same
dimension. If this dimension is $i$, by extension we say that this
separation is degenerated in dimension $i$.

\begin{lemma}~\label{lem:separation_box}
  Any separation $S$ of a proper $d$-tiling $\mathcal{T}$ is a
  $(d-1)$-dimensional box.
\end{lemma}
\begin{proof}
  This clearly holds for $d\le 2$, we thus assume that $d\ge 3$.
  Consider a separation $S$ that is degenerated in dimension $i$. By
  induction on $d$ we obtain that for any $j\in\{1,\ldots,d\}\setminus
  \{i\}$ and any $x\in [-1,+1]$ $S|^{(j)}_x$ is either empty, or it is
  a box. Indeed, for every generic $\mathcal{H}|^{(j)}_x$,
  $S|^{(j)}_x$ is either empty, or it is a separation of
  $\mathcal{T}|^{(i)}_x$.
  

  Let us first prove the lemma for $d=3$. If $S$ is not a
  $2$-dimensional box it has a $3\pi/2$ angle at some point $p$. It is
  clear that two boxes are necessary below (resp. above) $p$ with
  respect to the dimension $i$ to form the separation $S$, while the
  remaining $\pi/2$ angle at $p$ has to be covered by another (at
  least) fifth box, contradicting the fact that $\mathcal{T}$ is
  proper.

  For $d\ge 4$, the lemma follows from the following claim
  (considering $S\subseteq \mathcal{H}^{(i)}_{x} \simeq \RR^{d-1}$).
\begin{claim}
  Consider a connected set $S\subset \RR^d$ with $d\ge 3$, that is a
  finite union of $d$-dimensional boxes.  If for any $i\in \{1,\ldots
  ,d\}$ and any $x\in [-1,+1]$ we have that $S|^{(i)}_x$ is either
  empty or a box, then $S$ is a box.
\end{claim}
Towards a contradiction suppose that $S$ is not a box. In such case
for some dimension $i\in\{1,\ldots,d\}$, and for some values $x, x'\in
(-1,+1)$ the sets $S|^{(i)}_{x}$ and $S|^{(i)}_{x'}$ are two boxes
that differ (at least) on their $j^\text{th}$ interval for some $j\neq
i$. Thus following a curve going from $S|^{(i)}_{x}$ to
$S|^{(i)}_{x'}$ one goes through a point $p\in \RR^d$ with the
following property: For any sufficiently small $\epsilon$, the sets
$S|^{(i)}_{p_i}$ and $S|^{(i)}_{p_i+\epsilon}$ are two boxes that
differ (at least) on their $j^\text{th}$ interval for some $j\neq i$
and, as $S$ is a finite union of boxes, these boxes respectively
contain the points $p$ and $p'=(p_1,\ldots, p_{i-1},p_i+\epsilon,
p_{i+1}, \ldots,p_d)$. As their $j^\text{th}$ interval differ, we can
assume that for some $y\in (-1,+1)$ we have that $q=(p_1,\ldots, p_i,
\ldots , p_{j-1}, y, p_{j+1}, \ldots,p_d) \in S|^{(i)}_{p_i}$ and that
$q' = (p_1,\ldots, p_i+\epsilon, \ldots , p_{j-1}, y, p_{j+1},
\ldots,p_d) \notin S|^{(i)}_{p_i+\epsilon}$.  This implies that for
any $k\neq i,j$, the set $S|^{(k)}_{p_k}$ contains exactly three of
the aforementioned four points, contradicting the fact that it is a
box.\hfill\qed
\end{proof}

A $d$-tiling $\mathcal{T}$ is said in \emph{general position} if it
does not contain two coplanar separations, that is two separations
belonging to the same hyperplane $\mathcal{H}^{(i)}_{x}$.
\begin{lemma}~\label{lem:gen_position}
Any proper $d$-tiling $\mathcal{T}$ can be slightly perturbated to
obtain an equivalent one $\mathcal{T'}$ that is in general position.
Here equivalent means that the elements of $\mathcal{T}$ and
$\mathcal{T'}$ are in bijection and that two boxes of $\mathcal{T}\cup
\mathcal{T}_{ext}$ touch in dimension $i$ if and only if the
corresponding two boxes of $\mathcal{T'}\cup \mathcal{T}_{ext}$ touch
in dimension $i$.
\end{lemma}
\begin{proof}
First note that two coplanar separations do not intersect. Otherwise,
if two such separations $S$ and $S'$ would intersect at some point
$p$, then there would be a box $A$ below $S$ (with respect to $i$)
that would contain $p$ and a box $B$ above $S'$ (with respect to $i$)
that would also contain $p$. As $A$ and $B$ intersect they should
intersect on a $(d-1)$-dimensional box, but as this intersection is
degenerated in dimension $i$ then $S(A,i,+) \sim S(B,i,-)$,
contradicting the fact that they belong to distinct equivalence
classes.

Let us now proceed by induction on the number of separations that are
coplanar with another separation. Consider such a separation $S$ that
is degenerated in dimension $i$. For any box $A$ below $S$ (resp. $B$
above $S$) with respect to dimension $i$, replace $A_i =[A^-_i,A^+_i]$
with $[A^-_i,A^+_i+\epsilon]$ (resp. replace $B_i =[B^-_i,B^+_i]$ with
$[B^-_i -\epsilon, B^+_i]$). It is clear that for a sufficiently small
$\epsilon$ any two boxes intersect and touch on a dimension $j\neq i$
in $\mathcal{T'}\cup \mathcal{T}_{ext}$ if and only if they did in
$\mathcal{T}\cup \mathcal{T}_{ext}$. The intersections degenerated in
dimension $i$ either remained the same, either were simply translated
(for those belonging to $S$). As we have now one separation less that
is coplanar with another separation, we are done.\hfill\qed
\end{proof}

\section{Simplicial complexes and Dushnik-Miller dimension}\label{sec:dimdm}

Given a $d$-tiling $\mathcal{T}$, the simplicial complex
$\mathcal{S}(\mathcal{T})$ \emph{induced} by $\mathcal{T}$ is defined
as follows. Let $\mathcal{T}$ be the vertex set of
$\mathcal{S}(\mathcal{T})$ and let a set $F\subseteq \mathcal{T}$ be a
face of $\mathcal{S}(\mathcal{T})$ if and only if the elements of $F$
intersect, that is if $\bigcap_{B\in F} B \not=\emptyset$. From this
definition, it is clear that if $F\subseteq \mathcal{T}$ is a face of
$\mathcal{S}(\mathcal{T})$, any subset of $F$ is also a face of
$\mathcal{S}(\mathcal{T})$. So $\mathcal{S}(\mathcal{T})$ is indeed a
simplicial complex.

The Dushnik-Miller dimension of a simplicial complex $\mathcal{S}$,
denoted by $\dimdm(\mathcal{S})$, is the minimum integer $k$ such that
there exist $k$ linear orders $<_1,\ldots,<_k$ on $V$, where $V$ is
the vertex set of $\mathcal{S}$, such that for every face $F$ of
$\mathcal{S}$ and for every vertex $u\in V$, there exists some $i$
such that $\forall v\in F$, $v\le_i u$. Such set of orders is said to
be a \emph{realizer} of $\mathcal{S}$. Note that if $\mathcal{T}$ has
$p$ pairwise intersecting boxes, $\mathcal{S}(\mathcal{T})$ has a face
$F$ of size $p$ (usually such face is said to have dimension $p-1$),
and this implies that $\dimdm(\mathcal{S}(\mathcal{T})) \ge
p$. Indeed, every vertex $v\in F$ has to be greater than the other
vertices of $F$ in some order. This shows why Theorem~\ref{thm:main} is
tight in most cases.

\newcommand{\GT}[0]{\overrightarrow{G\:}\hspace{-1.5mm}\left(\mathcal{T}\right)}
\newcommand{\GTp}[0]{\overrightarrow{G\:}\hspace{-1.5mm}\left(\mathcal{T'}\right)}
\newcommand{\GTT}[2]{\overrightarrow{G\:}\hspace{-1.5mm}\left(\mathcal{T}|^{#1}_{#2}\right)}

\begin{definition}
  Given a proper $d$-tiling $\mathcal{T}$, let $\GT{}$ be the
  digraph with vertex set $\mathcal{T}$ and with an arc $AB$ if and only
  if there exist points $a\in A$ and $b\in B$ such that $b_i < a_i$
  for all $i\in\{1,\ldots,d\}$. Note that this is equivalent as
  saying that $\GT{}$ has an arc $AB$ if and only if $B^-_i < A^+_i$
  for all $i\in\{1,\ldots,d\}$.
\end{definition}

\begin{lemma}~\label{lem:acyclic_2}
  For any proper $d$-tiling $\mathcal{T}$, the digraph $\GT{}$ is
  acyclic.
\end{lemma}
\begin{proof}
We proceed by induction on $d$ and on the size of $\mathcal{T}$.  The
lemma clearly hold for $d=1$ or for $|\mathcal{T}|=1$, and we thus
focus on the induction step. We thus consider any proper $d$-tiling
$\mathcal{T}$ with $d>1$ and $|\mathcal{T}|>1$, and we assume that the
lemma holds for any $d'$-tiling with $d'<d$, and for any $d$-tiling
with less boxes than $\mathcal{T}$. By Lemma~\ref{lem:gen_position} we
also assume that $\mathcal{T}$ is in general position.

Let $X$ be the unique box of $\mathcal{T}$ containing the point
$(-1,\ldots,-1)$.  It is clear that $X$ is a sink in $\GT{}$, so it
suffices to show that $\GT{} \setminus X$ is acyclic. In the following
we show this by constructing a proper $d$-tiling $\mathcal{T'}$ with
one box less than $\mathcal{T}$ and such that $\GT{} \setminus X$ is a
subgraph of $\GTp{}$. The digraph $\GTp{}$ being acyclic by induction
hypothesis, so is its subgraph $\GT{} \setminus X$.

\begin{claim}
  There exists an $i\in \{1,\ldots,d\}$ and a box $Y\in \mathcal{T}$
  such that $X_i^+ = Y_i^-$, and such that $X_j^+ = Y_j^+$ for all
  $j\neq i$.
\end{claim}
For any point $p\in\mathbb{R}^d$, any $\epsilon>0$, and any box $B$,
the set $\mathcal{P}(p,\epsilon)$ of $2^d$ points
$(p_1+\epsilon_1,\ldots,p_d+\epsilon_d)$, for $\epsilon_i\in
\{-\epsilon,+\epsilon\}^d$, intersect $B$ on a number of points that
is a power of two. Furthermore, when $\epsilon$ is sufficiently small,
all the boxes of $\mathcal{T}$ intersecting $\mathcal{P}(p,\epsilon)$
contain the point $p$ (and thus intersect each other), and any point
of $\mathcal{P}(p,\epsilon)$ belongs to exactly one box.  Thus for the
point $p$ defined by $p_i =X_i^+$, and for a sufficiently small
$\epsilon>0$, the box $X$ contains exactly one point of
$\mathcal{P}(p,\epsilon)$. But as $|\mathcal{P}(p,\epsilon)|$ is even,
there is (at least) one other box in $\mathcal{T}$ that contains
exactly one point of $\mathcal{P}(p,\epsilon)$, let us denote $Y$ this
box. Since $X$ and $Y$ intersect on a $(d-1)$-dimensional box, let us
denote $i$ the dimension where they touch, and note that $X_i^+ =
Y_i^-$, and that $X_j^+ = Y_j^+$ for all $j\neq i$, otherwise $Y$
would intersect $\mathcal{P}(p,\epsilon)$ on more points.

\begin{claim}
  For any box $B\in \mathcal{T}$ touching $X$ in dimension $i$ we have
  that its side $S(B,i,-)$ is contained in $X$'s side $S(X,i,+)$.
\end{claim}
The previous claim implies that for each $j\neq i$ there exists a
separation containing $S(X,j,+)$ and $S(Y,j,+)$. By
Lemma~\ref{lem:separation_box} such separation contains the box
$[-1,X^+_1]\times \cdots \times [-1,X^+_{i-1}]\times
[-1,X^+_i+\epsilon]\times [-1,X^+_{i+1}]\times \cdots \times
[X^+_{j},X^+_{j}]\times\cdots \times [-1,X^+_{d}]$. If some box $B\in
\mathcal{T}$ touching $X$ in dimension $i$ has its side $S(B,i,-)$ not
contained in $X$'s side, for example because $S(B,i,-)_j \not\subseteq
S(X,i,+)_j$, then some interior point of $S(B,i,-)$ is also in the
interior of $S(X,j,+)$, a contradiction.

We thus define $\mathcal{T'}$ from $\mathcal{T}\setminus X$ in the
following way:
\begin{itemize}
  \item For any box $B\in \mathcal{T}\setminus X$ touching $X$ in
    dimension $i$ we define a box $B'$ in $\mathcal{T'}$ by setting
    $B'_i = [-1,B^+_i]$ and $B'_j=B_j$ for $j\neq i$.
\item Any other box $B\in \mathcal{T}\setminus X$ is contained in
  $\mathcal{T'}$. In this context this box is denoted $B'$.
\end{itemize}
\begin{claim}
  $\mathcal{T'}$ is a proper $d$-tiling.
\end{claim}
Every box $B'\in \mathcal{T'}$ contains the corresponding box $B\in
\mathcal{T}\setminus X$ so if there is a point $p'\in [-1,+1]^d$ not
covered by $\mathcal{T'}$, it is a point of $X$. But by construction
this would imply that the point $p$, defined by $p_i= X^+_i+\epsilon$
and $p_j=p'_j$ for any $j\neq i$, is not covered by $\mathcal{T}$, a
contradiction. One can similarly prove that the boxes of
$\mathcal{T'}$ are interior disjoint. $\mathcal{T'}$ is thus a
$d$-tiling. It remains to prove that it is a proper one. Towards a
contradiction, assume that there exist two intersecting boxes $A', B'
\in \mathcal{T'}$ that touch in at least two dimensions. As
$A\subseteq A'$ and $B\subseteq B'$ the boxes $A$ and $B$ do not
intersect. This implies that one of these boxes, say $A$, touches $X$
in dimension $i$ while the other, $B$, touches $X$ in a dimension
$j\neq i$. This implies that $A'$ and $B'$ touch in dimension $j$ and
in another dimension $k\in\{1,\ldots,d\}\setminus \{i,j\}$. We thus
either have that $A^+_k=A'^+_k = B'^-_k = B^-_k$ or that
$B^+_k=B'^+_k=A'^-_k = A^-_k$. Whatever the case we denote $x$ this
value, and as both $A_k$ and $B_k$ intersect $X_k$ in more than one
point, we have that $-1<x<A^+_k$. As $\mathcal{T}$ is in general
position, it admits a separation $S$ such that $S_k=[x,x]$ containing
a point in the interior of $S(X,i,+)$ (as it is bordered by $A$) and a
point in the interior of $S(X,j,+)$ (as it is bordered by $B$). As $S$
is a box it thus contains a point in $X$'s interior, a
contradiction.

As any box $B\in \mathcal{T} \setminus X$ is contained in the
corresponding box $B'\in \mathcal{T'}$, for any arc $AB\in \GT{}
\setminus X$ we have that $B'^-_i \le B^-_i < A^+_i = A'^+_i$ for all
$i\in\{1,\ldots,d\}$, and thus there is an arc $A'B'\in \GTp{}$.
This concludes the proof of the lemma.
\hfill\qed
\end{proof}

\begin{theorem}~\label{thm:main}
Given a proper $d$-tiling $\mathcal{T}$, we have that
$\dimdm(\mathcal{S}(\mathcal{T}))\leq d+1$.
\end{theorem}
\begin{proof}
Consider the orders $(<_1,\ldots,<_{d+1})$ defined as follows. If two
distinct boxes $A, B\in \mathcal{T}$ are such that $B_i^- < A_i^+$ for
all $1\le i\le d$, then $A <_{d+1} B$. By Lemma~\ref{lem:acyclic_2},
the transitive closure of this relation is antisymmetric. In the
following let $<_{d+1}$ be any of its linear extensions.  For $<_i$
with $1\le i\le d$, given two boxes $A, B\in \mathcal{T}$, $A <_{i} B$
if and only if $A_i^- < B_i^-$, or if $A_i^- = B_i^-$ and $A <_{d+1}
B$. Those relations are clearly total orders.

Let us now prove that $\{ <_1,\ldots,<_{d+1} \}$ is a realizer of
$\mathcal{S}(\mathcal{T})$. To do so, we prove that for any point $p$
and any box $B\in \mathcal{T}$, that the set $\mathcal{A}(p)$ of boxes
containing $p$ is dominated by $B$ in some order $<_i$. By extension
of notation, in such case we say that $\mathcal{A}(p) <_i B$.
Note that $B$ verifies one of the following cases:
\begin{itemize}
\item[(1)] $p_i < B_i^-$ for some $1\le i\le d$,
\item[(2)] $B_i^- < p_i$ for all $1\le i\le d$, or
\item[(3)] $B_i^- = p_i$ for all $i\in I$, for some non-empty set
  $I\subseteq \{1,\ldots, d\}$, and $B_i^- < p_i$, otherwise.
\end{itemize}

In case (1), as $A_i^- \le p_i < B_i^-$ for all $A\in\mathcal{A}(p)$,
then $A <_i B$ for all $A\in \mathcal{A}(p)$, that is $\mathcal{A}(p)
<_i B$.

In case (2), as $B^-_i < p_i \le A_i^+$ for all $A\in\mathcal{A}(p)$
and all $1\le i\le d$, then $\mathcal{A}(p) <_{d+1} B$.

Case (3) is more intricate. Towards a contradiction we suppose that
$\mathcal{A}(p) \not\leq_i B$ for all $1\le i\le d+1$. Let $I\subseteq
\{1,\ldots, d\}$ be the non-empty set such that $B_i^- = p_i$ if and
only if $i\in I$.  Consider the directed graph $D$ with vertex set $I$
and which contains an arc $(i,j)$ if and only if there exists a box
$A\in \mathcal{A}(p)$ such that $A_i^-=p_i$ and $A_j^+=p_j$. Note that
as every box is $d$-dimensional $D$ has no loop $(i,i)$.
\begin{claim}
Every vertex in $D$ has at least one outgoing arc.
\end{claim}
For any $i\in I$, as $\mathcal{A}(p) \not\leq_i B$, there exists at
least one box $A \in \mathcal{A}(p)$ such that $B <_i A$. By
definition of $<_i$ this box is such that $A_i^-=p_i$ (as $p_i = B_i^-
\le A_i^- \le p_i$) and such that $B <_{d+1} A$. As $A \not<_{d+1} B$
there exists a $j\in \{1,\ldots,d\}$ such that $A_j^+ \le B_j^-$, but
as $B_j^- \le p_j \le A_j^+$, we have $A_j^+ = p_j = B_j^-$. Thus
$j\in I$ and $D$ has an arc $(i,j)$.

Thus $D$ is not acyclic and we can consider a circuit of minimum
length $C=(i_0,\ldots,i_k)$ in $D$, with $k\ge 1$.  For every $j\in
\{0,\ldots , k\}$ let $A{(j)}$ be a box of $\mathcal{A}(p)$ such that
$A(j)^-_{i_j}=p_{i_j}$ and $A(j)^+_{i_{j+1}}=p_{i_{j+1}}$ (where
$j+1$ is understood modulo $k+1$).
\begin{claim}
All the $k+1$ boxes $A{(j)}$ are distinct.
\end{claim}
Towards a contradiction, consider there
exists two distinct values $j$ and $j'$ such that $A{(j)}$ and $A{(j')}$
are identical. Let us call $A$ this box. By definition of $A{(j)}$
and $A{(j')}$, this box $A$ is such that $A_{i_j}^-=p_{i_j}$
$A_{i_{j'+1}}^+=p_{i_{j'+1}}$. Thus there is an arc from $i_j$ to
$i_{j'+1}$, contradicting the minimality of $C$.

Note that the intersection of these $k+1$ boxes is degenerate in
dimensions $i_j$ for $0\le j \le k$. Thus these $k+1$ boxes intersect
in a box of dimension at most $d-k-1$. This contradicts
Theorem~\ref{thm:prop=d+1-k} and concludes the proof of the theorem.
\hfill\qed
\end{proof}

\section{Conclusion}\label{sec:ccl}

It would be appreciable to deal with contact system of boxes
instead of $d$-tilings, that is to deal with sets of interior disjoint
$d$-dimensional boxes not necessarily spanning $[-1,+1]^d$ or
$\mathbb{R}^d$. For this purpose, we conjecture the following.
\begin{conjecture}
  A set $\mathcal{C}$ of $d$-dimensional boxes in $[-1,+1]^d$ is a
  subset of a proper $d$-tiling $\mathcal{T}$ if and only if every set
  $\mathcal{B}\subset \mathcal{C}$ of pairwise intersecting boxes is
  such that $\dim(\cap_{B\in \mathcal{B}} B)=d+1-|\mathcal{B}|$.
\end{conjecture}

Thomassen~\cite{Tho} (see also~\cite{Fusy}) characterized the
intersection graphs of proper $2$-tilings, exactly as the strict
subgraphs of the 4-connected planar triangulations. The 4-connected
planar triangulations are those where every triangle bounds a face. A
simplicial complex has the \emph{Helly property} if every clique in
its skeleton is a face of the simplicial complex. As the simplicial
complexes defined by $d$-tilings have the Helly property, we conjecture
the following:

\begin{conjecture}
  A simplicial complex $\mathcal{S}$ is such that $\mathcal{S} =
  \mathcal{S}(\mathcal{T}\cup\mathcal{T}_{ext})$ for some proper
  $d$-tiling $\mathcal{T}$, if and only if $\mathcal{S}$ is a
  triangulation of the $d$-dimensional octahedron with the Helly
  property and with Dushnik-Miller dimension $d+1$.
\end{conjecture}
  
Similarly, is it possible to generalize the fact that bipartite planar
graphs are the intersection graphs of non-intersecting and axis
parallel segments in the plane~\cite{FPP} ?

\subsubsection*{Acknowledgments}
The second author thanks \'E. Fusy for explaining him the construction
in~\cite{Z10} defining a Schnyder wood from a 2-tiling. This was
crucial for obtaining this result.

\bibliographystyle{abbrv}

\begin{thebibliography}{10}

\bibitem{AA10}
  A. Asinowski and T. Mansour,
  \newblock Separable d-Permutations and Guillotine Partitions.
  \newblock {\em Annals of Combinatorics} 14(1): 17--43, 2010.

\bibitem{EuroCG}
W. Evans, S. Felsner, S.G. Kobourov, and T. Ueckerdt,
\newblock  Graphs admitting $d$-realizers: spanning-tree-decompositions and box-representations.
\newblock {\em Proc. of EuroCG '14}. 

\bibitem{FF}
S. Felsner and M.C. Francis,
\newblock Contact representations of planar graphs with cubes.
\newblock {\em Proc. of SoCG '11}, 315--320, 2011. 

\bibitem{FK}
S. Felsner and S. Kappes,
\newblock Orthogonal Surfaces and Their CP-Orders.
\newblock {\em Order} 25: 19-47, 2008. 

\bibitem{FPP}
H. de Fraysseix, J. Pach, and R. Pollack,
\newblock How to draw a planar graph on a grid.
\newblock {\em Combinatorica} 10(1): 41--51, 1990.

\bibitem{Fusy}
\'E. Fusy,
\newblock Transversal structures on triangulations: A combinatorial study and straight-line drawings.
\newblock {\em Discrete Math.} 309(7), 1870--1894, 2009.

\bibitem{GI17}
  D. Gon\c{c}alves, and L. Isenmann,
  \newblock Dushnik-Miller dimension of TD-Delaunay complexes.
  \newblock \emph{EuroGC '17}, 2017.
  
\bibitem{Graham-Pollak}
R.L. Graham, and H.O. Pollak,
\newblock On embedding graphs in squashed cubes.
\newblock {\em Graph theory and applications}, Lecture Notes in Math. 303: 99--110, 1972.

\bibitem{Oss}
P. Ossona de Mendez,
\newblock Geometric realization of simplicial complexes.
\newblock Proc. of Graph Drawing '99, {\em LNCS} 1731: 323--332, 1999. 

\bibitem{S89}
W. Schnyder,
\newblock Planar graphs and poset dimension.
\newblock {\em Order} 5: 323--343, 1989. 

\bibitem{Tho} 
C. Thomassen, 
\newblock Plane representations of graphs.
\newblock Progress in graph theory (Bondy and Murty, eds.), 336--342, 1984.

\bibitem{Tverberg}
H.~Tverberg,
\newblock On the decomposition of $K_n$ into complete bipartite graphs.
\newblock {\em J.  Graph Theory}, 6: 493--494, 1982.

\bibitem{Z85}
J.~Zaks,
\newblock How Does a Complete Graph Split into Bipartite Graphs and How are Neighborly Cubes Arranged ?
\newblock {\em Amer. Math. Monthly}, 92(8): 568--571, 1985.

\bibitem{Z10}
  H.~Zhang,
  \newblock Planar Polyline Drawings via Graph Transformations.
  \newblock \emph{Algorithmica}, 57: 381--397, 2010.

\end{thebibliography}


\end{document}